\documentclass[smallextended]{svjour3}
\usepackage{fix-cm}
\usepackage{bbm}
\usepackage{amsmath}
\usepackage{amssymb}
\usepackage{url}
\usepackage{longtable}

\journalname{}
\begin{document}
\title{Quantum digital-to-analog conversion algorithm using decoherence}
\titlerunning{Quantum digital-to-analog conversion algorithm using decoherence}

\author{Akira SaiToh}
\institute{A. SaiToh \at
Department of Computer Science and Engineering, Toyohashi University of Technology,
1-1 Hibarigaoka, Tenpaku-cho, Toyohashi, Aichi 441-8580, Japan\\
\email{saitoh@sqcs.org}}

\date{Received: date / Accepted: date}

\maketitle

\begin{abstract}
We consider the problem of mapping digital data encoded on a quantum register
to analog amplitudes in parallel. It is shown to be unlikely that a fully
unitary polynomial-time quantum algorithm exists for this problem; NP becomes
a subset of BQP if it exists. In the practical point of view, we propose
a nonunitary linear-time algorithm using quantum decoherence. It tacitly uses an
exponentially large physical resource, which is typically a huge number of identical
molecules. Quantumness of correlation appearing in the process of the algorithm is also
discussed.
\end{abstract}



\section{Introduction}\label{sec_intro}
There have been many conventional electrical and algorithmic implementations for
digital-to-analog conversion (DAC) \cite{Rad11}. For a very brief explanation, given
digital data $f(k)\in\{0,1\}^n$ for label $k\in\{0,1\}^m$, DAC produces an analog
signal with an amplitude proportional to $f(k)$ for each $k$. (We do not give a
particular physical meaning to the variable $k$ although it is often a label of a
time slot. It can be regarded as a label of a channel especially in the context of
parallel processing.)
The DAC problem can be formulated as follows.
\begin{definition}\label{def_DAC}
{\bf\tt DAC problem}\\
{\bf Instance}: Integers $m, n \ge 1$, set $\{k\}$ of integers $k\in\{0,1\}^m$, and
function $f: \{0,1\}^m\rightarrow \{0,1\}^n$.\\
{\bf Limitation}: Function $f$ can be used at most once for each $k$ only for the purpose of
preparing digital data $f(k)$.\\
{\bf Task}: Generate analog amplitudes $V(k)\propto f(k)$ for all $k$
(here, the proportionality constant is a certain constant independent of $k$).
\end{definition}
A classical parallel solver for this problem can be intuitively implemented; one may prepare the
same converter for each $k$ although this is not economical. In view of recent development
of quantum information processing, it is natural to consider the possibility of using
quantum parallelism instead of preparing many identical converters. Then the parallelization
will be quite economical. We are seeking for this possibility in this contribution, but we
will first notice that it is not straight-forward to find a useful quantum
algorithm for DAC.

Quantum parallelism plays an essential role in quantum computing \cite{G99,NC2000}.
In a conceptual explanation, a unitary operation
\[
U_f:|0\rangle|k\rangle\mapsto |f(k)\rangle|k\rangle
\]
corresponding to a function $f: \{0,1\}^m\rightarrow\{0,1\}^n$ is often considered
(here, $k\in \{0,1,\ldots,2^m-1\}$). It is tacitly assumed that $U_f$ is constructed
with a reasonable number of single-qubit and two-qubit unitary operations. Once
$U_f$ is constructed, we may apply it to a superposition of quantum states by 
linearity of a unitary transformation.
Given the initial state $|0\rangle\sum_{k=0}^{2^m-1}c_k|k\rangle$ with complex
amplitudes $c_k$ satisfying $\sum_k|c_k|^2=1$, the following evolution can be performed:
\begin{equation}
|0\rangle\sum_kc_k|k\rangle\overset{U_f}{\mapsto}\sum_kc_k|f(k)\rangle|k\rangle.
\end{equation}
The resultant state is a sort of a digital state because $f(k)$ is kept as a digital
datum in the left register. Thus this process is an encoding of digital data
as a quantum state. For simplicity, we may often choose $c_k=1/\sqrt{2^m}$. To
perform DAC in a quantum manner, we need to change the amplitude $c_k$ so that $|c_k|^2$
is proportional to $f(k)$ for each $k$ in parallel. This will however turn out to be a
hard problem in the context of computational cost, later in this paper.
Furthermore, it is not easy to retrieve the amplitude distribution as classical signals
in a normal quantum computing model since a measurement destroys the state; this is
another problem.

There is a conventional quantum algorithm for analog-to-digital and digital-to-analog
conversions (ADC and DAC) in a different context: Schm\"user and Janzing \cite{SJ05}
considered the problem of conversions between a continuous wave function $\psi(x)$
($x\in[0,L]$ stands for a position) and a quantum state $|\Psi\rangle
=\frac{1}{\sqrt{2^n}}\sum_{j=0}^{2^n-1}\psi(jL/2^n)|j\rangle$ of an $n$-qubit
register. They employed the Jaynes-Cummings model as a physical model and showed
procedures for the conversions using several physical time evolutions that are possible
under the model. Their algorithm, however, requires an exponentially large interaction
strength or an exponential interaction time. Besides, it is probably unnatural to
call their conversions as ADC and DAC because $|\Psi\rangle$ has amplitude information
in its complex amplitudes. They are rather conceptually close to analog-to-analog
conversions in this sense.

Here, we are considering the DAC problem along with Definition~\ref{def_DAC} and going to
construct a certain quantum algorithm for it. In this context, there have been several
related works: Ventura and Martinez \cite{VM99} considered a problem to generate a superposition
with desired phase factors for the complex amplitudes that have unit absolute values. (Thus their
problem is a phase modulation problem.) They developed an algorithm whose cost scales quasi-linearly
in the number of components of a superposition; this cost is exponential in the number of
qubits. Other related works are those about decomposition of an arbitrary unitary transformation
\cite{T99,V04,NKS06,Sh06,LRY13}. (They are applicable to constructing a unitary DAC.)
They all require exponential number of quantum gates for decomposition. These conventional results
suggest that it is unlikely to achieve exponential speedup of DAC by using quantumness as a
resource. This is very natural in the sense that {\tt NP} will be included in {\tt BQP} if
such a speedup is possible by using a fully unitary process. It is also unlikely for any
parallel computer to perform a parallel DAC within polynomial resources. These points will be
discussed in Sec.~\ref{sec_complexity}. (See, {\em e.g.}, Ref.~\cite{W08} for quantum
computational complexity classes, such as {\tt BQP}.) In addition, use of these conventional
works does not resolve the other problem of how to retrieve classical signals from a resultant
quantum state.

In this contribution, we develop a nonunitary quantum algorithm for the DAC problem. We can
encode a given instance as a quantum state in Hilbert space, while we move to Liouville space
(the operator Hilbert space) in subsequent steps in the algorithm. We may also start with a
mixed state from the first step. We utilize decoherence processes in addition to standard
quantum gates to construct the algorithm. The resultant mixed state possesses all of the DAC
output data with pointer addresses. We introduce a nondemolition retrieval of desired signals,
which relies on the ensemble property of a physical system. The run time of our algorithm is
linear in the input size, while the physical resource it consumes is exponential. It should be
noted that quantum computing in Liouville space is not a very rare topic. See, {\em e.g.},
Refs.~\cite{Br00,Ta02}.

As our model is not a standard quantum computer model, it is also of interest to
consider quantumness of computation. Although we do not have a direct measure for it,
we will discuss quantumness of correlation between registers, in particular, quantum
discord \cite{OZ2001}.

This paper is organized as follows. We first prove the hardness of the parallel
DAC problem in Sec.~\ref{sec_complexity}. Then we introduce our algorithm in
Sec.~\ref{sec_alg}. Complexity of the algorithm is also discussed in the section.
In addition, quantumness of correlation appearing in the algorithm is argued
in Sec.~\ref{sec_quantumness}. Section~\ref{sec_con} gives a summary and remarks
on the results.

\section{Computational difficulty of parallel DAC}\label{sec_complexity}
As mentioned in Sec.~\ref{sec_intro}, it is unlikely to achieve an exponential
speedup of DAC using any parallel computer. This is because the following proposition
holds.
\begin{proposition}\label{prop_chard}
Consider the DAC problem. Assume that {\em i.i.d.} noise appear in the DAC outputs and
the noise intensity for each output $V(k)$ is bounded above by a constant $\varepsilon_{\rm th}$.
Suppose we can achieve the amplitudes $V(k)$ for all $k$ within ${\rm poly}(m, n)$ time and
${\rm poly}(m, n)$ physical resource on average for any $f$. (The internal circuit of DAC may
possess a probabilistic behavior. This is why we discuss the average cost.) Then SAT is
solvable within polynomial cost on average.
\end{proposition}
This is easily proved:
\begin{proof}
Consider a conjunctive normal form (CNF) $\varphi(x_0,\ldots,x_{m-1})$ with
variables $x_i\in\{0,1\}$. Let us write $k=x_0\cdots x_{m-1}$ and choose $f$
such that $f(k)=\varphi(k)_{n-1}0_{n-2}0_{n-3}\cdots0_{0}$ (the subscripts are bit labels).
We choose $n$ so that $2^{n-1}V_{\rm LSB} > 2^m (\varepsilon_{\rm th}+c)$, where $V_{\rm LSB}$
stands for the least-significant-bit voltage (namely, the smallest nonzero output amplitude
of a DAC device) and $c>0$ is a certain small constant\footnote{It is a tacit assumption that
any DAC is designed to satisfy $V_{\rm LSB}\gtrapprox\varepsilon_{\rm th}$ \cite{devDAC,Rad11}.
Thus, setting $n$ to a value slightly larger than $m$ should be enough to satisfy the inequality
$2^{n-1}V_{\rm LSB} > 2^m (\varepsilon_{\rm th}+c)$. Even if the tacit assumption does not
hold, a value of $n$ enough to satisfy the inequality scales linearly in $m$ as long as
$\varepsilon_{\rm th}$ is a constant.}.
Prepare the set $\{0,\ldots,2^m-1\}$ for $\{k\}$.
Then, after DAC is complete for all $k$, we look at the mixture of resultant amplitudes $V(k)$.
Apart from noise, for each $k$, a nonzero analog signal is generated for nonzero $\varphi(k)_{n-1}$ only.
This signal corresponds to the most significant bit [namely, the ($n-1$)th bit] so that its
intensity is $2^{n-1}V_{\rm LSB}$. In contrast, the total noise intensity of the mixture is
obviously bounded above by $2^m \varepsilon_{\rm th}$. Therefore, the signal-per-noise ratio is
$>1+c/\varepsilon_{\rm th}$, which has a constant gap from $1$. Hence, on average, a constant number
of accumulation of data of the mixture is enough to decide if $\varphi$ is satisfiable. This
is owing to the central limit theorem in regard with the sum of random variables \cite{Shiryaev,Klenke}
(this holds by assumption of {\em i.i.d.} noise).
\end{proof}
Here we have considered the hardness of DAC in the context of general parallel
computation. It is also possible to consider it in the context of quantum computation,
as we will discuss next.

\subsection{Case of fully unitary quantum process}
Although we do not employ a fully unitary process in our algorithm to be shown in
Sec. \ref{sec_alg}, here it is meaningful to consider a quantum DAC described as a
unitary map. This will show the hardness of quantum DAC in the context of standard quantum
computation. Let us consider a map $\hat{V}$ having the following property as a unitary map
realizing a quantum DAC.
\[
\hat{V}: (1/\sqrt{N_k})\sum_k|f(k)\rangle|k\rangle
         \mapsto \sum_k \hat{f}(k)|f(k)\rangle|k\rangle
\]
where $N_k$ is the size of set $\{k\}$ and $\hat{f}(k)$'s are complex amplitudes
satisfying the conditions $|\hat{f}(k)|^2\propto f(k)+\epsilon$ and
$\sum_k |\hat{f}(k)|^2=1$. Here, $\epsilon < 1/2^m$ is a small bias; without it,
$\hat{V}$ becomes nonunitary when $f(k)$ is always zero. (Of course, such an
$\epsilon$ is negligible in case there is nonzero $f(k)$.) Under these settings, the
existence of a quantum circuit realizing $\hat{V}$ is manifest: one can use a known
method \cite{Sh06} to implement an arbitrary unitary operation although this requires
an exponential number of basic quantum gates in general.

Let us now show that the following proposition holds, which indicates the hardness
of implementing $\hat{V}$. (Here, it should be recalled that $k$ is an $m$-bit string and
$f(k)$ is an $n$-bit string. As for the definition of {\tt BQP}, see, {\em e.g.},
Ref.~\cite{W08}.)
\begin{proposition}\label{prop_uhard}
Suppose there exists a quantum circuit realizing $\hat{V}$ with circuit depth and
circuit width ${\rm poly}(m, n)$ for any $f$. Then SAT is included in {\tt BQP}.
\end{proposition}
\begin{proof}
The proof is conceptually similar to the previous one.
Consider a CNF $\varphi:\{0,1\}^m\rightarrow\{0,1\}$ and choose $f$ such that
$f(k)=\varphi(k)_{n-1}0_{n-2}0_{n-3}\cdots0_{0}$ where $k$ is an $m$-bit string.
Let us choose $n\ge 2$.
First, we prepare the initial state $|0\rangle|0\rangle$ where the left register
consists of $n$ qubits and the right register consists of $m$ qubits.
Second, we apply $H^{\otimes m}$ to the right register where $H=(|0\rangle\langle0|
+|0\rangle\langle1|+|1\rangle\langle0|-|1\rangle\langle1|)/\sqrt{2}$ is an
Hadamard transform. The state becomes $(1/\sqrt{2^m})|0\rangle\sum_{k=0}^{2^m-1}|k\rangle$.
Third, because $f$ is implementable as a classical boolean circuit, it is straight-forward
to transform it to a quantum circuit using the right register as its input and the left
register as its output. This circuit changes the state into
$(1/\sqrt{2^m})\sum_{k=0}^{2^m-1}|f(k)\rangle|k\rangle$. As a convention, the internal cost
to apply $f$ is not considered since it is a sort of an oracle.
Fourth, we apply $\hat{V}$ to the state and obtain $\sum_k \hat{f}(k)|f(k)\rangle|k\rangle$.
Then we measure the right register. There are two cases:
(i) In case $\varphi$ is satisfiable,
among $m$-qubit states $|k\rangle$, those dissatisfying $\varphi$ individually have relative
population $< 1/2^{m+n-1}$ by the assumption on $\hat{f}(k)$. Thus the probability
$p_{\rm s}$ to obtain a bit string satisfying $\varphi$ by measuring the right register
is larger than $1/(1+\mu)$ with $\mu=(2^m-1)/2^{m+n-1}$. Because we chose $n\ge 2$, we
have $p_{\rm s}>2/3$.
\linebreak (ii) In case $\varphi$ is not satisfiable,
measuring the right register produces a uniformly random number from $\{k\}$.
Fifth, we use a classical circuit to verify the resultant $m$-bit string.
Case (i) results in ``yes'' with probability $>2/3$ and case (ii) results in
``no'' with probability $1$. By definition of {\tt BQP}, the proof is now completed.
\end{proof}
Here is an alternative statement of the proposition:
\begin{corollary}\label{cor_qhard}
There exist functions $f$ for which there is no quantum circuit to implement $\hat{V}$
with circuit depth and circuit width ${\rm poly}(m, n)$ unless ${\tt NP}\subseteq{\tt BQP}$.
\end{corollary}

\section{Nonunitary algorithm}\label{sec_alg}
We have seen the hardness of parallel DAC in the previous section.
As it has turned out to be unlikely that a polynomial-time unitary quantum algorithm
exists, we seek for the possibility of polynomial-time nonunitary quantum algorithm.
As it has also turned out to be unlikely to have a polynomial-cost parallel
algorithm in a more general sense, we may compromise with increase in physical resource.
In particular, a massive parallelism of an ensemble of identical systems like a molecular
ensemble is considered a realistic resource in our strategy. Our algorithm is entirely
constructed as a completely-positive trace-preserving linear map (a CPTP map) and
hence physically feasible (see, {\em e.g.}, Refs.~\cite{BZ06,HG12} for the CPTP condition).

\subsection{Decoherence maps}
We first revisit two decoherence maps that are very common in quantum physics
\cite{NC2000,HG12}. These maps are used in our algorithm.
\begin{definition}
For probability $p$ and a $d\times d$ density matrix $\rho$, the {\em depolarization map}
$\Lambda_{\rm dpl}$ is defined as
\[
\Lambda_{\rm dpl}: \rho\mapsto (1-p)\rho + p({\rm Tr} \rho)I/d
\]
with $I$ the $d\times d$ identity matrix. We also write
\[
\Lambda_{\rm dpl}(p, \rho) = (1-p)\rho + p({\rm Tr} \rho)I/d.
\]
\end{definition}
\begin{definition}
For probability $p$ and a $d\times d$ density matrix $\rho$, the {\em dephasing map}
$\Lambda_{\rm ph}$ is defined as
\[
\Lambda_{\rm ph}: \rho\mapsto
(1-p)\rho + p\left(\sum_{i=0}^{d-1}\langle i|\rho|i\rangle|i\rangle\langle i|\right).
\]
We also write
\[
\Lambda_{\rm ph}(p, \rho)
= (1-p)\rho + p\left(\sum_{i=0}^{d-1}\langle i|\rho|i\rangle|i\rangle\langle i|\right).
\]
\end{definition}

These map are physically realizable in practice, in particular in magnetic
resonance systems like NMR and ESR, and linear optical systems. In magnetic
resonance systems, there is a long history of controlling spin relaxations
\cite{Meh2001}. Artificial dephasing is achievable by a gradient-field pulse
applied to a target spin together with decoupling pulses applied to other
spins \cite{XJ06}. Artificial depolarization is achievable by using the same
method in X, Y, and Z bases appropriately (see \cite[Sec. 8.3]{NC2000} for
the decomposition of the depolarizing error). Another way to realize the maps
is to use natural $T_1$ and $T_2$ relaxations \cite{Becker80} together with
decoupling pulses whose strength and durations are calibrated to achieve desirable
relaxations. In addition, artificial dephasing can also be achieved by randomly
flipping an ancillary spin coupled to a target spin as demonstrated by Kondo
{\em et al.} \cite{Kondo07}.

In linear optical systems too, it is possible to perform artificial decoherence
operations by using common optical devices. Suppose we employ polarization qubits.
Then, artificial depolarization can be made by a Sagnac-type interferometer
\cite{Alm07,Jeo13}. Artificial dephasing can also be achieved by using the birefringence
of an optical fiber \cite{Berg00,RoB13}. (See also Refs.~\cite{Ai07,La08} for other
ways to perform mixed-state processing in linear optics.)

Besides to the above decoherence maps, we use the unit sample function given by
\[
\delta(x)=\left\{
\begin{array}{ll}1&~\;(x=0)\\0&~\;(x\not = 0)\end{array}\right.
\]
to describe our algorithm.

\subsection{Algorithm construction}
In the following, we are going to present our algorithm.
First, we introduce a subroutine beforehand.\\
~\\
{\bf Subroutine {\tt S\_1}:}~\\
{\bf Input:} An $(m+n)$-qubit pure state
\[
|\psi_{fk}\rangle=|f_{n-1}f_{n-2}\ldots f_0\rangle^{\rm L}|k\rangle^{\rm R},
\]
where $f_i\in\{0,1\}$ ($i=0,\ldots,n-1$) and $k\in\{0,1\}^m$. Superscripts ${\rm L}$
and ${\rm R}$ stand for left and right registers.\\
{\bf Output:} $(m+2n)$-qubit mixed state
\[
\begin{split}
\tilde{\rho}_{fk}
=&|f_{n-1}\rangle\langle f_{n-1}|\otimes\cdots
\otimes|f_0\rangle\langle f_0|\otimes|k\rangle\langle k|\\
&\otimes\left[\left(\frac{I}{2}\right){\delta(f_{n-1})}+(|0\rangle\langle0|){\delta(f_{n-1}-1)}\right]\\
&\otimes\cdots
\otimes\left[\left(\frac{I}{2}\right){\delta(f_{0})}+(|0\rangle\langle0|){\delta(f_{0}-1)}\right].
\end{split}
\]
{\bf Construction:}
\begin{itemize}
\item[(i)] Attach $n$ ancillary qubits in the state $|0_{n-1}\cdots 0_0\rangle$ to the input.
\item[(ii)] For each $i$, select $|f_i\rangle$ and the corresponding ancilla qubit $|0_i\rangle$
to make a pair and apply the $0$-controlled-Hadamard gate (${\rm C_0H}$) with the control
bit $f_i$ and the target bit $0_i$.
Here, ${\rm C_0H}=|0\rangle\langle0|\otimes H + |1\rangle\langle1|\otimes I$.
The resultant state of each pair is
$|f_i\rangle\left[(|+\rangle){\delta(f_i)}+(|0\rangle){\delta(f_i-1)}\right]$ with
$|+\rangle=(|0\rangle+|1\rangle)/\sqrt{2}$.
\item[(iii)] For each qubit of register ${\rm R}$, apply the fully dephasing map
$\Lambda_{\rm ph}(p=1)$.
\item[(iv)] For each $i$, apply the fully dephasing map $\Lambda_{\rm ph}(p=1)$ to the
qubits in the pair made in (ii) individually. The resultant state of each pair is
$|f_i\rangle\langle f_i|\otimes\left[\left(\frac{I}{2}\right){\delta(f_{i})}
+(|0\rangle\langle0|){\delta(f_{i}-1)}\right]$.
\end{itemize}

One can regard Subroutine {\tt S\_1} as a map
\[
\Lambda_1: \mathcal{H}_2^{\otimes m+n}\rightarrow \mathcal{S}(\mathcal{H}_2^{\otimes m+2n}),
\]
where $\mathcal{H}_2$ is a two-dimensional Hilbert space and generally
$\mathcal{S}(\mathcal{H})$ is a space of density matrices (namely, a state space) acting
on Hilbert space $\mathcal{H}$. $\Lambda_1$ maps $|\psi_{fk}\rangle$ to $\tilde{\rho}_{fk}.$
Then, it is straightforward to find the following fact. 
\begin{remark}\label{rem_Lambda_1}
Map $\Lambda_1$ has the property:
\[
\Lambda_1(|\psi_{fk}\rangle+|\psi_{f'k'}\rangle)
=\Lambda_1(|\psi_{fk}\rangle)
+\Lambda_1(|\psi_{f'k'}\rangle),
\]
where $|\psi_{f'k'}\rangle = |{f'}_{n-1}\cdots {f'}_0\rangle^{\rm L}|k'\rangle^{\rm R}$.
\end{remark}
This is because, for $k \not = k'$, step (iii) erases the terms with factors $|k\rangle^{\rm R}\langle k'|$
and its Hermitian conjugates in a density matrix represented in the computational basis.

One can also regard Subroutine {\tt S\_1} as a map
\[
\widetilde{\Lambda}_1:\mathcal{S}(\mathcal{H}_2^{\otimes m+n})
\rightarrow \mathcal{S}(\mathcal{H}_2^{\otimes m+2n}).
\]
It maps $|\psi_{fk}\rangle\langle\psi_{fk}|$ to $\tilde{\rho}_{fk}.$
We can easily reach the following fact simply by linearity of the operations
used in the subroutine.
\begin{remark}
Map $\widetilde{\Lambda}_1$ has the property:
\[
\widetilde{\Lambda}_1(|\psi_{fk}\rangle\langle\psi_{fk}|+|\psi_{f'k'}\rangle\langle\psi_{f'k'}|)
=\widetilde{\Lambda}_1(|\psi_{fk}\rangle\langle\psi_{fk}|)
+\widetilde{\Lambda}_1(|\psi_{f'k'}\rangle\langle\psi_{f'k'}|).
\]
\end{remark}

Let us now introduce the main routine of our algorithm.\\~\\
{\bf Algorithm {\tt NONUNITARY\_QUANTUM\_DAC}:}\\
{\bf Input:} Integers $m, n \ge 1$ and function $f: \{0,1\}^m\rightarrow \{0,1\}^n$.\\
{\bf Output:} A mixture of component states, labeled by $k$, each of which possesses
a signal of analog amplitude $V(k)\propto f(k)$. [The explicit form of this mixture is found
as Eq.~(\ref{eq_rho_2}) or equivalently as Eq.~(\ref{eq_rho_2}') in the construction.
$V(k)$ can be individually derived from the mixture by using one of appropriate routines
defined later.]\\
{\bf Construction:}
\begin{enumerate}
\item Generate either of the pure state
\begin{equation}\label{eq_psi_o}
|\psi_o\rangle=(1/\sqrt{2^m})\sum_{k=0}^{2^m-1}|f(k)\rangle^{\rm L}|k\rangle^{\rm R}
\end{equation}
or the mixed state
\begin{equation}\label{eq_rho_o}
\rho_o=(1/2^m)\sum_{k=0}^{2^m-1}|f(k)\rangle^{\rm L}\langle f(k)|\otimes|k\rangle^{\rm R}\langle k|.
\end{equation}
Here, $|\psi_o\rangle$ can be generated from
$(I^{\otimes n}\otimes H^{\otimes m})
|0_{n-1}\cdots0_0\rangle^{\rm L}|0_{m-1}\cdots0_0\rangle^{\rm R}
=(1/\sqrt{2^m})\sum_{k}|0\rangle^{\rm L}|k\rangle^{\rm R}$. As $f$ is a (multiple-bit-output) boolean
function, it is easy to construct a quantum circuit $C_f$ mapping each $|0\rangle^{\rm L}|k\rangle^{\rm R}$ 
to $|f(k)\rangle^{\rm L}|k\rangle^{\rm R}$.\\
$\rho_o$ can be generated from $|0\rangle^{\rm L}\langle0|\otimes (I/2)^{\otimes m}
=(1/2^m)\sum_k |0\rangle^{\rm L}\langle0|\otimes |k\rangle^{\rm R}\langle k|$.
By linearity of a quantum circuit, $C_f$ maps each
$|0\rangle^{\rm L}\langle0|\otimes|k\rangle^{\rm R}\langle k|$ to
$|f(k)\rangle^{\rm L}\langle f(k)|\otimes|k\rangle^{\rm R}\langle k|$.
\item Apply Subroutine {\tt S\_1} to the state (either $|\psi_o\rangle$ or $\rho_o$).
In either of the cases, the resultant state is
\begin{equation}\label{eq_rho_1}\begin{split}
\rho_1=
\frac{1}{2^m}\sum_{k=0}^{2^m-1}\biggl\{&
|f_{n-1}(k)\rangle\langle f_{n-1}(k)|\otimes\cdots
\otimes|f_0(k)\rangle\langle f_0(k)|\otimes|k\rangle\langle k|\\
&\otimes\left[\left(\frac{I}{2}\right){\delta(f_{n-1}(k))}+(|0\rangle\langle0|){\delta(f_{n-1}(k)-1)}\right]^{{\rm a}_{n-1}}\\
&\otimes\cdots\otimes
\left[\left(\frac{I}{2}\right){\delta(f_{0}(k))}+(|0\rangle\langle0|){\delta(f_{0}(k)-1)}\right]^{{\rm a}_0}\biggr\},
\end{split}\end{equation}
where $f_i(k)$ is the $i$th bit of $f(k)$ and ${\rm a}_i$ stands for the $i$th ancilla
qubit.
\item For each $i\in\{0,\ldots,n-1\}$, apply $\Lambda_{\rm dpl}(p=p_i)$ to ${\rm a}_i$
with $p_i=1-2^{i-n+1}$. Let us write $q_i=2^{i-n+1}$.
The resultant state is
\begin{equation}\label{eq_rho_2}\begin{split}
\rho_2=
\frac{1}{2^m}\sum_{k=0}^{2^m-1}\biggl\{&
\left[|f_{n-1}(k)\rangle\langle f_{n-1}(k)|\otimes\cdots
\otimes|f_0(k)\rangle\langle f_0(k)|\right]^{\rm L}\otimes|k\rangle^{\rm R}\langle k|\\
&\otimes
\biggl[\left(\frac{I}{2}\right){\delta(f_{n-1}(k))}+
p_{n-1}\left(\frac{I}{2}\right){\delta(f_{n-1}(k)-1)}\\
&\hspace{7mm}+q_{n-1}(|0\rangle\langle0|){\delta(f_{n-1}(k)-1)}\biggr]^{{\rm a}_{n-1}}\\
&\otimes\cdots\otimes
\biggl[\left(\frac{I}{2}\right){\delta(f_{0}(k))}+
p_{0}\left(\frac{I}{2}\right){\delta(f_{0}(k)-1)}\\
&\hspace{15mm}+q_{0}(|0\rangle\langle0|){\delta(f_{0}(k)-1)}\biggr]^{{\rm a}_{0}}
\biggr\}.
\end{split}\end{equation}
This is the output of the algorithm. To clarify this output as a mixture of DAC outputs,
now we introduce a nonstandard representation inspired by the deviation density matrix
representation \cite{Jones2010}. Let us consider a deviation from $I/2$ for an ancilla
qubit state because $I/2$ does not produce a signal for any traceless observable.
We use $|\hat{Z}\rangle\!\rangle\equiv Z/2 + I/2$ (here $Z=|0\rangle\langle0|-
|1\rangle\langle1|$ is the Pauli Z operator) and neglect $I/2$ for each ancilla qubit.
We also use $|0\rangle\!\rangle\equiv |0\rangle\langle0|$ and $|1\rangle\!\rangle\equiv
|1\rangle\langle1|$ for each qubit of registers $L$ and $R$. Then the output state can be
rewritten as
\begin{equation}\tag{\ref{eq_rho_2}'}
|{\hat{\rho}}_2\rangle\!\rangle=\frac{1}{2^m}\sum_{k=0}^{2^m-1}\biggl\{
\left[\bigotimes_{i=0}^{n-1}|f_i(k)\rangle\!\rangle\right]^{\rm L}
\otimes |k\rangle\!\rangle^{\rm R}
\otimes \bigotimes_{i=0}^{n-1}2^{i-n+1}\delta(f_i(k)-1)|\hat{Z}\rangle\!\rangle^{{\rm a}_i}
\biggr\}.
\end{equation}
With this representation, it is clear that the output state possesses a mixture of DAC
outputs in its ancilla register and each DAC output is pointed by the address (or pointer)
kept in the ${\rm R}$ register. In order for fetching a specified analog amplitude from this
resultant state, either {\tt FETCH\_SIGNAL\_1} or {\tt FETCH\_SIGNAL\_2} described in
Sec.~\ref{secfetch} should be used.\footnote{\label{foot_human}Note that fetching one resultant
datum takes at least one step for a human in any parallel processing model. On the other hand,
resultant data are usable for subsequent parallel processing. In fact, $\rho_2$ can be directly
used for further parallel processing within the ensemble quantum computing model.}
\end{enumerate}

\subsection{Signal fetching}\label{secfetch}
As mentioned in the algorithm, we introduce two routines to fetch a specified analog amplitude
from $\rho_2$. First we introduce the following routine involving a selective projection. The other
one introduced later does not involve it.\\~\\
{\bf Routine {\tt FETCH\_SIGNAL\_1}:}\\
{\bf Input:} Density matrix $\rho_2$ output from {\tt NONUNITARY\_QUANTUM\_DAC} and integer $k$.\\
{\bf Output:} Analog signal derived from the $k$th component of $\rho_2$.\\
{\bf Construction:}\begin{itemize}
\item[~] As an observable, let us employ $\sum_{i=0}^{n-1}Z_i$ where each Pauli Z operator
$Z_i$ acts on the corresponding ancilla qubit ${\rm a}_i$.
The desired analog amplitude is obtained by
\[
V(k)={\rm Tr}\left[(P_k^{\rm R}\rho_2P_k^{\rm R})(I^{\rm LR}\otimes\sum_{i=0}^{n-1}Z_i)\right],
\]
where $P_k=|k\rangle\langle k|$ is a projector. (Here, we do not normalize the trace of
$P_k^{\rm R}\rho_2P_k^{\rm R}$. This is because the signal intensity of a state is
proportional to its physical population in practice.)
This equation indicates that we project $\rho_2$ onto the subspace of
$|\cdot\rangle|k\rangle|\cdot\rangle$ and perform the ensemble- or time-average measurements
of polarizations $Z_i$ in order to obtain $V(k)$. It is straightforward to find that
\begin{equation}\label{eq_Vfin}
\begin{split}
V(k) &= \frac{1}{2^m}\sum_{i=0}^{n-1}q_i\delta(f_i(k)-1)\\
     &\propto f(k).
\end{split}
\end{equation}
\item[$*$] Comment: we do not assume that $P_k^{\rm R}$ is a real projection. Any method to pick up a
signal from a component state possessing $|k\rangle^{\rm R}\langle k|$ can be used. In addition, we
assume that this routine can be used repeatedly. See the discussion below.
\end{itemize}
Alternatively, one can use the following routine. It looks more natural because it does not contain
a selective projection.\\~\\
{\bf Routine {\tt FETCH\_SIGNAL\_2}:}\\
{\bf Input:} Density matrix $\rho_2$ output from {\tt NONUNITARY\_QUANTUM\_DAC} and integer $k$.\\
{\bf Output:} Analog signal derived from the $k$th component of $\rho_2$.\\
{\bf Construction:}\begin{itemize}
\item[(i)] Apply $|k\rangle^{\rm R}$-controlled $\bigotimes_{i=0}^{n-1}H^{{\rm a}_i}$,
namely, the Hadamard transform acting on the ancilla register controlled by the condition that
register ${\rm R}$ is in the state $|k\rangle$.
\item[(ii)] Perform an ensemble- or time-average measurement using the observable $\sum_{i=0}^{n-1} X_i$,
where $X_i$ is a Pauli X operation ($|0\rangle\langle1|+|1\rangle\langle0|$) acting on the $i$th ancilla
qubit. 
\item[(iii)] Perform the same operation as (i) for state recovery.
\item[$*$] Comment: Because the state of each ancilla qubit in $\rho_2$ has only diagonal elements,
the output signal at step (ii) is originated from the sole component affected by step (i).
Using the well-known relation $HXH=Z$, it is easy to find that the amplitude of this signal is given
by Eq.~(\ref{eq_Vfin}). In addition, this routine is also assumed to be usable repeatedly.
\end{itemize}

Now we discuss physical feasibility of the above routines for signal fetching. As for the projection
$P_k^{\rm R}$ used in {\tt FETCH\_SIGNAL\_1}, it can be regarded as any operation to activate the $k$th
component for readout. Typically this is done by adjusting the frequency of a readout pulse when we
use a sort of magnetic resonance techniques \cite{M78,F91,Fo14}. Alternatively, one may apply some
operation to set only the target component ready for readout ({\em e.g.}, by setting it in resonance with a
readout pulse). In this sense, step (i) of {\tt FETCH\_SIGNAL\_2} is indeed a kind of such an operation,
and hence the two routines are conceptually same. As for the nondemolition measurement requisite for
repeated use of the routines, we simply utilize an ensemble property.
For example, in a molecular spin ensemble system, many identical copies of the same system can be
simultaneously input into the algorithm. Therefore, it is valid to assume that we may have a mixture of
many identical copies of $\rho_2$ in the output. Then both of the routines can be used repeatedly since a
measured signal is a macroscopic property in this case. See, {\em e.g.}, Ref.~\cite{KML11} for validity of
this claim.
For another example, consider a linear optical quantum system in which vertical and horizontal polarizations
of a photon are used as single-qubit basis states $|0\rangle$ and $|1\rangle$, respectively. Suppose we
inject many identical photonic states massively and repeatedly \cite{Peters03} in the input of the algorithm
in the way that they are noninteracting with each other, using time-bin and/or frequency-bin separations
(see, {\em e.g.}, the introductory parts of Refs.~\cite{Or06,Zh08,Ra09,U10} for these separations).
Then we have an photonic ensemble in the state $\rho_2$ in the output (output repeatedly). Although
{\tt FETCH\_SIGNAL\_1} is not usable\footnote{
Photons are not physically connected like molecular spins; a projection using {\em e.g.} polarizing
beam splitters for register ${\rm R}$ does not separate signals in the ancilla register.},
{\tt FETCH\_SIGNAL\_2} works and can be used repeatedly as is clear by its structure.\footnote{Here, we
are discussing an algorithmic structure. It is, of course, a formidable challenge to implement a large
quantum circuit with linear optics \cite{MW12}.}
Among these two examples, the former is conceptually more natural as we may regard the molecular
ensemble as a static memory keeping the DAC outputs.

Furthermore, there is a variant of {\tt FETCH\_SIGNAL\_1} with which we can derive multiple $V(k)$'s at once,
when the system we employ is a bulk-ensemble molecular spin system with a magnetic resonance facility. There is
a strategy called the fetching algorithm \cite{XL02,LX03}. Consider the case where the qubits in the
ancilla register ${\rm A}$ are decoupled to each other, while they are individually coupled to the remaining
part ${\rm LR}$. Then, we measure each ancilla qubit ${\rm a}_i$ by a free-induction-decay (FID) measurement.
The Fourier spectrum (or the FID spectrum) of each measurement exhibits splitting of peaks (each subpeak
corresponds to each $k$). The magnitude of each subpeak in the FID spectrum of ${\rm a}_i$ corresponds to
$v_i(k)={\rm Tr} (P_kZ_i)\rho_2$. Each $V(k)$ is obtained by $V(k)=\sum_{i=0}^{n-1} v_i(k)$. A drawback
is that there are possibly so many nonzero $v_i(k)$'s that corresponding subpeaks are not clearly
separated. The spectrum resolution for this strategy needs to be exponentially fine in general.

In addition to the normal demand for retrieving individual DAC outputs, one may need to test if there
is any nonzero DAC output. Then one has only to measure the ancilla qubits of the output state $\rho_2$
with the observable $\sum_{i=0}^{n-1}Z_i$ instead of calling the above routines. In particular, it is
enough to measure the $(n-1)$th ancilla qubit with $Z$ when one uses the algorithm for solving an unsorted
search problem by choosing the qubit as an oracle qubit. Then it works as a variant of the Br\"uschweiler's
bulk-ensemble search \cite{Br00,SK06} although this is not a mainly intended usage of our algorithm. The
measurement here should also be an average measurement supported by an exponentially large physical
resource. Therefore hardness of the parallel DAC problem shown in Proposition~\ref{prop_chard} is
unchanged, while our algorithm runs in linear time.

In a similar point of view, signal mixing of multiple DAC outputs is also possible by slightly changing the
above routines. To retrieve a mixed signal of $r$ signals, one can either replace the projection in
{\tt FETCH\_SIGNAL\_1} with a rank-$r$ projection for $r$ $|k\rangle$'s or introduce multiple
$|k\rangle^{\rm R}$-controlled $\bigotimes_{i=0}^{n-1}H^{{\rm a}_i}$ operations for $r$ $|k\rangle$'s
instead of just one in {\tt FETCH\_SIGNAL\_2}.

Finally, the following things should be emphasized again.
Our algorithm has been constructed as a nonunitary algorithm using decoherence processes.
Thus the output state is a mixed state even if the initial state is a pure state.
Among the registers ${\rm L}$, ${\rm R}$, and ${\rm A}$, the ${\rm R}$ register works as an address
register and the ${\rm A}$ register works as a data register where DAC outputs are placed in the end
of the algorithm. To fetch a specified analog amplitude, we need to use an average measurement
on the ${\rm A}$ register subsequent to a choice of an address. 

We now turn into a little specific explanation of our algorithm for clarity and then
evaluate the computational cost.

\subsection{A little specific explanation}
It will clarify the behavior of our algorithm if we focus on the evolution of
a single component state. Consider the case of $n=3$. We begin with the end of
step 1 of the algorithm. Let us pick up $|101\rangle^{\rm L}|\cdot\rangle^{\rm R}$ among
the component states $|f_2(k)f_1(k)f_0(k)\rangle^{\rm L}|k\rangle^{\rm R}$. Here, we have assumed
that one of $f_2f_1f_0$'s is $101$.

In step 2, Subroutine {\tt S\_1} is applied. In this subroutine, an ancilla register
$|0^{\rm a_2}0^{\rm a_1}0^{\rm a_0}\rangle^{\rm A}$
is attached and the ${\rm C_0H}$ operation is applied to each pair of $f_i$ and ${\rm a_i}$.
This changes the component state we are tracking to
$|101\rangle^{\rm L}|\cdot\rangle^{\rm R}|0+0\rangle^{\rm A}$.
Then the dephasing operations are applied in this subroutine.
The component state of our interest evolves into $(|101\rangle^{\rm L}\langle101|
\otimes|\cdot\rangle^{\rm R}\langle\cdot|\otimes[|0\rangle^{\rm a_2}\langle0|
\otimes (I/2)^{\rm a_1}\otimes|0\rangle^{\rm a_0}\langle0|]^{\rm A}$.

Then, in step 3, depolarizing operations (with individually different parameter
values) are applied to the ancillary qubits. The component state we are tracking becomes
$(\cdots)^{\rm LR}\otimes\left\{|0\rangle^{\rm a_2}\langle0|\otimes(I/2)^{\rm a_1}
\otimes [(1-1/2^2)(I/2)+1/2^2|0\rangle\langle0|]^{\rm a_0}\right\}^{\rm A}$.

Finally, we consider the signal fetching. In case we employ {\tt FETCH\_SIGNAL\_1},
the average measurement is performed with the observable $Z_2+Z_1+Z_0$ together with
a projection onto a component state.
For the component state we are tracking, this results in the output signal
${\rm Tr}Z_2 |0\rangle^{\rm a_2}\langle0| + {\rm Tr}Z_1 (I/2)^{\rm a_1}
 + {\rm Tr}Z_0 [(1-1/2^2)(I/2)+1/2^2|0\rangle\langle0|]^{\rm a_0}
= 1\times(1/2^0)+0\times(1/2^1)+1\times(1/2^2)$. This is proportional to the
value $101$.
In case we employ {\tt FETCH\_SIGNAL\_2}, first a controlled Hadamard operation is applied
so that the ancilla qubits of our target component are transformed by the Hadamard
operation. The component state evolves into $(\cdots)^{\rm LR}\otimes
\left\{|+\rangle^{\rm a_2}\langle+|\otimes(I/2)^{\rm a_1}
\otimes [(1-1/2^2)(I/2)+1/2^2|+\rangle\langle+|]^{\rm a_0}\right\}^{\rm A}$.
Then the average measurement is performed with the observable $X_2+X_1+X_0$.
This results in the output signal
${\rm Tr}X_2 |+\rangle^{\rm a_2}\langle+| + {\rm Tr}X_1 (I/2)^{\rm a_1}
 + {\rm Tr}X_0 [(1-1/2^2)(I/2)+1/2^2|+\rangle\langle+|]^{\rm a_0}
= 1\times(1/2^0)+0\times(1/2^1)+1\times(1/2^2)$, which is same as that of the above
case. After achieving this desired signal, the controlled Hadamard operation is
applied again to restore the state.

It is easy to consider the evolution of any other component state among
$|f_2(k)f_1(k)f_0(k)\rangle^{\rm L}|k\rangle^{\rm R}$'s.
The above explanation has been for the case of $n=3$, but larger $n$ does not
make much difference in the explanation.

So far we have introduced our algorithm and also given a specific explanation.
We will next discuss the computational cost of our algorithm.

\subsection{Computational cost}
It should be firstly mentioned that in general each quantum operation including decoherence
acts on its target (sub)system for all the components at once in the ensemble quantum
computing model \cite{NC2000,Jones2010}. (This is different from classical parallel computing
where usually one operation acts on one target component.) We analyze computational costs of
our algorithm {\tt NONUNITARY\_QUANTUM\_DAC} on the basis of this standpoint.

The run time of our algorithm is $O(m+n)$ for all the steps except for the process
of applying $C_f$ in the first step. Thus the dominant factor is the circuit depth of
$C_f$ implementing function $f$. This is however hidden behind a query cost in the
convention of computational complexity theory \cite{AB09}. Hence we can state that
the time cost of our algorithm is $O(m+n)$.

As for space, we use only $m+2n$ qubits apart from those used inside $C_f$. Although
$C_f$ may use a certain number of ancilla qubits internally, this is also hidden
behind a query cost. Thus we can state that our algorithm uses $m+2n$ qubits.
It should be noted that counting qubits appearing in the algorithm is quite superficial
as this does not reflect the resource needed for average measurements.

As for the query cost, $f$ is called only once as a function acting on a superposition
or a mixed state. Thus we use only one query in our algorithm.

The dominant cost in the use our algorithm is obviously the cost for an average measurement
to fetch a specified result from the resultant ensemble. Since each intensity has the factor
$1/2^m$ as shown in Eq.~(\ref{eq_Vfin}), exponentially many identical data are needed to obtain
the signal intensity in the presence of noise. In case noise is a random one, accumulation of
$L$ data results in the signal per noise ratio
$\propto L/2^m:\sqrt{L}$. Hence $\mathcal{O}(2^{2m})$ identical data are needed to overcome noise.
Thus we need a bulk ensemble of $\mathcal{O}(2^{2m})$ identical systems for ensemble averaging
or $\mathcal{O}(2^{2m})$ data acquisitions for time averaging.

In addition, in case {\tt FETCH\_SIGNAL\_2} is employed for data fetching, the controlled
Hadamard operations should be constructed, for which another $m$ ancilla qubits and $O(m+n)$ time
are sufficient. This cost applies for each time we fetch a signal from the resultant ensemble.

Despite the drawback of the measurement cost, it is still feasible for physical
implementation when we utilize a massively parallel molecular system, such as those for
NMR/ESR \cite{Meh2001}, for up to $10$ or a little more qubits in the present technology.
Theoretically, it is physically feasible when $m\lesssim (\log_2 10^{23})/2\simeq 38$
and clean qubits are initially available by the use of very low temperature or an
efficient state-initialization technique \cite{SV99,B02,KKN03,SK05}.

\section{Quantumness of correlation in the algorithm}\label{sec_quantumness}
It is often discussed if a quantum correlation like entanglement \cite{PV06,H09review} and
quantum discord \cite{OZ2001} is a source of computational power of quantum computers although
there is no definite answer presently \cite{H09review,Br11}.

In our algorithm, there is clearly a large amount of entanglement for non-constant $f$ 
({\em i.e.}, $f$ such that $f(k)$ is not same for all $k$) when we choose a pure-state process
by employing $|\psi_o\rangle$ in step 1. The controversial case is when we choose a mixed-state
process by employing $\rho_o$ in step 1. Then, there is neither entanglement nor quantum discord
in the states $\rho_o$ [Eq.~(\ref{eq_rho_o})], $\rho_1$ [Eq.~(\ref{eq_rho_1})], and
$\rho_2$ [Eq.~(\ref{eq_rho_2})] because they are written in the form of a sum of products
of diagonal states. There is, however, a certain quantumness in the process: In
Subroutine {\tt S\_1}, firstly the ${\rm C_0H}$ operation is applied to each pair of $f_i$ and
its corresponding ancilla qubit. This operation changes the ancilla register state to
$|a(k)\rangle^{\rm A}=\bigotimes_i(|+\rangle_i\delta(f_i(k))+|0\rangle_i\delta(f_i(k)-1))$.
The entire state at this point is
\begin{equation}\label{eq_tilrho}
\widetilde{\rho}=(1/2^m)\sum_{k=0}^{2^m-1}|f(k)\rangle^{\rm L}\langle f(k)|
\otimes|k\rangle^{\rm R}\langle k|\otimes|a(k)\rangle^{\rm A}\langle a(k)|.
\end{equation}
It does not have entanglement as it is a convex combination of product states. It may possess,
however, quantum discord because $|+\rangle$ and $|0\rangle$ are nonorthogonal to each other.
(Historically, this kind of mixed state, namely, a mixture of pure states locally possessing
X-basis eigenstates and those locally possessing Z-basis eigenstates, has been used as a very
typical example \cite{HV01,HKZ04} to exhibit nonzero quantum discord).
Quantumness appears at this point only, in the mixed-state process. Note that $\widetilde{\rho}$
is invariant under partial transpose so that any method using a spectrum change caused by partial
transpose \cite{SRN12} cannot be used for quantifying quantumness of correlation in the present case.

In the following, we investigate quantumness of correlation in $\widetilde{\rho}$ by calculating
quantum discord under a certain setting simplified for readability of equations. Let us begin
with a brief explanation of the definition of quantum discord.

Briefly speaking, for bipartite system $\mathcal{AB}$, quantum discord is defined as a discrepancy
between two different forms of mutual information, which are equivalent in classical regime \cite{OZ2001}.
One of the forms is $\mathcal{I}(\mathcal{A}:\mathcal{B})
=S(\mathcal{A})+S(\mathcal{B})-S(\mathcal{AB})$ where $S(\mathcal{AB})$
is the von Neumann entropy\footnote{In general, the von Neumann entropy of a density matrix $\rho$
is calculated as $S(\rho)=-{\rm Tr}\rho\log_2\rho=-\sum_\lambda \lambda\log_2\lambda$ where
$\lambda$'s are the eigenvalues of $\rho$.} of the density matrix $\rho^\mathcal{AB}$ of $\mathcal{AB}$
and $S(\mathcal{A})$ ($S(\mathcal{B})$) is that of the reduced density matrix $\rho^\mathcal{A}$
($\rho^\mathcal{B}$) of $\mathcal{A}$ ($\mathcal{B}$). 
It should be noted that we employ the base-two logarithm when calculating entropies in this paper.
The other form is $\mathcal{J}(\mathcal{A}:\mathcal{B})
=S(\mathcal{B})-S(\mathcal{B}|\mathcal{A})$ where
$S(\mathcal{B}|\mathcal{A})=\sum_j p_j S(\mathcal{B}|j)$ is a conditional entropy with $p_j$ the
probability that event $j$ appears on $\mathcal{A}$. It is clear that $S(\mathcal{B}|\mathcal{A})$
depends on the measurement we choose, in general for a quantum system\footnote{
Thus, the theory of quantum discord has a focus on the measurement stage. This is in contrast to the
entanglement theory, which has a focus on the state generation stage \cite{PV06,H09review}.
}.
As these two forms are equivalent in classical regime, their minimum discrepancy over all available
measurements is regarded as an amount of quantum correlation.

Following the original formulation of Ref.~\cite{OZ2001}, the measurement is assumed to
be a simple von Neumann-type measurement acting on $\mathcal{A}$, represented as a complete set
of orthogonal one-dimensional projectors, $\Pi^\mathcal{A}=\{\Pi^\mathcal{A}_j\}$.
Then, we have $S(\mathcal{B}|\mathcal{A})=S(\rho^\mathcal{AB}|\Pi^\mathcal{A})$
where $S(\rho^\mathcal{AB}|\Pi^\mathcal{A})=\sum_jp_jS(\sigma_j)$ with
$p_j={\rm Tr}(\Pi^\mathcal{A}_j\otimes I^\mathcal{B}\rho^\mathcal{AB})$ and
$\sigma_j=(\Pi^\mathcal{A}_j\otimes I^\mathcal{B}\rho^\mathcal{AB}\Pi^\mathcal{A}_j\otimes I^\mathcal{B})/p_j$.

On the basis of the above equations, quantum discord \cite{OZ2001} is written as
\begin{equation}\label{eq_discord}
\begin{split}
D_{\underline{\mathcal{A}}}(\rho^\mathcal{AB})
&\equiv \underset{\Pi^\mathcal{A}}{\rm min}\left[
\mathcal{I}(\mathcal{A}:\mathcal{B})-\mathcal{J}(\mathcal{A}:\mathcal{B})\right]\\
&= S(\rho^\mathcal{A})-S(\rho^\mathcal{AB})
   +\underset{\Pi^\mathcal{A}}{\rm min}S(\rho^\mathcal{AB}|\Pi^\mathcal{A}),
\end{split}
\end{equation}
where we put subscript $\underline{\mathcal{A}}$ for indicating the subsystem the measurement acts on.
In general, quantum discord depends on which subsystem is measured, {\em i.e.}, the values of
$D_{\underline{\mathcal{A}}}(\rho^\mathcal{AB})$ and $D_{\underline{\mathcal{B}}}(\rho^\mathcal{AB})$
are often different. It should be mentioned that, although one may consider a positive operator-valued
measurement (POVM) instead of the simple one (see, {\em e.g.}, Refs.~\cite{HV01,HKZ04}),
the present definition is also common \cite{OZ2001,Luo08,Ali10}.

Our interest is quantum discord for $\widetilde{\rho}$ given in Eq.~(\ref{eq_tilrho}). For
simplicity, let us consider the bipartite splitting between ${\rm LR}$ and ${\rm A}$. Then,
the quantities $D_{\underline{\rm LR}}(\widetilde{\rho})$ and
$D_{\underline{\rm A}}(\widetilde{\rho})$ are of our concern.
For an explicit calculation, we are going to specify $f$ and perform a numerical computation to find the
minimum in Eq.~(\ref{eq_discord}). Here, a rather simple case is considered for the sake of readability
of resultant equations. Let us consider a clock function $f_{\rm c}$ such that $f_{\rm c}(k)=0
\equiv 0_{n-1}\cdots0_0$ for a half of $k$'s and $f_{\rm c}(k)=2^{n-1}\equiv 1_{n-1}0_{n-2}\cdots0_0$
for the remaining $k$'s, for $f$. We write $\widetilde{\rho}$ under this condition as
$\widetilde{\rho}_{\rm c}$.

First we calculate the von Neumann entropy for $\widetilde{\rho}_{\rm c}$ and reduced density
matrices $\widetilde{\rho}_{\rm c}^{\rm LR}$ and $\widetilde{\rho}_{\rm c}^{\rm A}$. It is straightforward
to obtain
\[
S(\widetilde{\rho}_{\rm c}) = S(\widetilde{\rho}_{\rm c}^{\rm LR}) = m.
\]
It is also easy to calculate $S(\widetilde{\rho}_{\rm c}^{\rm A})$:
\[\begin{split}
S(\widetilde{\rho}_{\rm c}^{\rm A})
&=S\left(\frac{|+\rangle\langle+|+|0\rangle\langle0|}{2}\otimes|+\cdots +\rangle\langle +\cdots +|\right)
 =S\left(\frac{|+\rangle\langle+|+|0\rangle\langle0|}{2}\right)\\
&=-\sum_{\pm}\lambda_{\pm}\log_2\lambda_{\pm}
\end{split}\]
with $\lambda_{\pm}=\frac{1}{2}\pm\frac{\sqrt{2}}{4}$. These results are used in later calculations.

We are now going to find the values of $D_{\underline{\rm LR}}(\widetilde{\rho}_{\rm c})$ and
$D_{\underline{\rm A}}(\widetilde{\rho}_{\rm c})$. The easier one is
$D_{\underline{\rm LR}}(\widetilde{\rho}_{\rm c})$. We have
\[
D_{\underline{\rm LR}}(\widetilde{\rho}_{\rm c})=S(\widetilde{\rho}_{\rm c}^{\rm LR})-S(\widetilde{\rho}_{\rm c})
+\underset{\Pi^{\rm LR}}{\rm min}S(\widetilde{\rho}_{\rm c}|\Pi^{\rm LR})
=\underset{\Pi^{\rm LR}}{\rm min}S(\widetilde{\rho}_{\rm c}|\Pi^{\rm LR}).
\]
This is obviously $\ge 0$ by the property of von Neumann entropy and it is possible to
achieve zero: Consider computational basis vectors $|lr\rangle$ of ${\rm LR}$ and employ
projectors $\Pi_{lr}^{\rm LR}=|lr\rangle\langle lr|$. Then $S(\widetilde{\rho}_{\rm c}|\Pi^{\rm LR})$
vanishes because, for each $lr$, the projected state
$(\Pi^{\rm LR}_{lr}\otimes I^{\rm A}\widetilde{\rho}_{\rm c}\Pi^{\rm LR}_{lr}\otimes I^{\rm A})/p_{lr}$
is a pure state. Consequently, \[D_{\underline{\rm LR}}(\widetilde{\rho}_{\rm c})=0.\]
This result is actually trivially obtained if we use \cite[Proposition 3]{OZ2001}, considering the
form of $\widetilde{\rho}_{\rm c}$.
Although the discord just calculated has vanished, the other one is shown to be nonzero as below.

We turn into the calculation of $D_{\underline{\rm A}}(\widetilde{\rho}_{\rm c})$. As we know
\begin{equation}\label{eq_discord_A}
\begin{split}
D_{\underline{\rm A}}(\widetilde{\rho}_{\rm c})
&=S(\widetilde{\rho}_{\rm c}^{\rm A})-S(\widetilde{\rho}_{\rm c})
  +\underset{\Pi^{\rm A}}{\rm min}S(\widetilde{\rho}_{\rm c}|\Pi^{\rm A})\\
&=-\sum_{\pm}\lambda_{\pm}\log_2\lambda_{\pm}-m
  +\underset{\Pi^{\rm A}}{\rm min}S(\widetilde{\rho}_{\rm c}|\Pi^{\rm A}),
\end{split}
\end{equation}
now we need to compute $\underset{\Pi^{\rm A}}{\rm min}S(\widetilde{\rho}_{\rm c}|\Pi^{\rm A})$.
Because there is no phase factor in the amplitudes of component states in $\widetilde{\rho}_{\rm c}$,
we can choose $\Pi^{\rm A}=\{|\phi_1\rangle\langle\phi_1|\otimes|+\cdots +\rangle\langle +\cdots +|,
|\phi_2\rangle\langle\phi_2|\otimes|+\cdots +\rangle\langle +\cdots +|\}$
with $|\phi_1\rangle=\cos\theta|0\rangle+\sin\theta|1\rangle$ and
$|\phi_2\rangle=\sin\theta|0\rangle-\cos\theta|1\rangle$, and perform a minimization
over $\theta$. After a tedious but straightforward calculation (see the supplementary material),
we find
\[
\begin{split}
\underset{\Pi^{\rm A}}{\rm min}S(\widetilde{\rho}_{\rm c}|\Pi^{\rm A})
=m + \underset{\theta}{\rm min}\biggl[\!&
-H\!\left(\!\frac{(\cos\theta+\sin\theta)^2}{4}\!+\!\frac{\cos^2\theta}{2}\!\right)\\
&+\frac{1}{2}H\!\left(\!\frac{(\cos\theta+\sin\theta)^2}{2}\!\right)+\frac{1}{2}H\!\left(\cos^2\theta\right)
\biggr],
\end{split}
\]
where $H(x)=-x\log_2 x-(1-x)\log_2(1-x)$ is the binary entropy function.
We used a simple numerical search to find the minimum. The result is
${\rm min}_{\Pi^{\rm A}}S(\widetilde{\rho}_{\rm c}|\Pi^{\rm A})\approx m-0.399124$.
In addition, we have $-\sum_{\pm}\lambda_{\pm}\log_2\lambda_{\pm} \approx 0.600876$.
Substituting these values into Eq.~(\ref{eq_discord_A}), we obtain
\[
D_{\underline{\rm A}}(\widetilde{\rho}_{\rm c})\approx 0.201752.
\]

Thus, we have found a nonzero quantum discord. Although there is no entanglement when we choose
the mixed-state process from the first in our algorithm, a certain quantum correlation still exists.
In the above calculation, we have chosen a simple clock function for $f$ for simplicity, but in
general nonvanishing quantum discord is involved owing to the nonorthogonality of $|+\rangle$
and $|0\rangle$ when $f$ is a nonconstant function.

\section{Concluding remarks}\label{sec_con}
We have considered a parallel DAC problem handling digital data generated by a
function $f:\{0,1\}^m\rightarrow\{0,1\}^n$ ($m, n \ge 1$ are integers).
This problem has been proven to be a difficult problem under a natural assumption
on physical noise, as shown in Proposition \ref{prop_chard} by reducing SAT to the
problem. A unitary quantum DAC has also been proven to be hard in the sense that a
polynomial-size quantum DAC circuit cannot be constructed for certain instances
unless {\tt NP} is a subset of {\tt BQP}. Thus it is unlikely to find an algorithm to
perform parallel DAC using a resource scaling polynomially in $m$ and $n$ in
physically feasible computational models.

On the basis of the above observation, we have developed an algorithm to perform a parallel
DAC using nonunitary processes. In particular, depolarization and dephasing channels
have been utilized. The output state given by Eq.~(\ref{eq_rho_2}) [or, equivalently,
Eq.~(\ref{eq_rho_2}')] is a mixture of component states possessing DAC output data and their
pointers. One can fetch each analog amplitude from the state by using a selective choice of a
pointer and an average measurement. The developed algorithm runs with linear time, linear
space, and a single query apart from the cost for fetching each amplitude from its output.
As mentioned in footnote \ref{foot_human}, the number of amplitude-fetching calls is not important.
It is usual for any parallel processing that picking up one item from the output takes at least
one step, while the output can be directly used for further parallel processing if we continue
to use the same system. This is true for our algorithm.

A drawback of our algorithm is the required physical resource: we need either a massive
bulk-ensemble system or a many-time accumulation to derive DAC outputs as visible signals.
With this drawback, we have discussed that, theoretically, our algorithm is still feasible for
implementation using a molecular spin system when $m\lesssim 38$.

We have also discussed on quantumness of correlation between registers used in our
algorithm. As we have seen, we may choose either a pure-state process or a mixed-state
process for the first part of our algorithm. Entanglement exists for the former case, while
it does not for the latter case. The mixed-state process, however, involves quantum discord.
Hence our algorithm is quantum even when a mixed-state process is employed in the sense
that it uses a quantum state possessing quantum correlation.

Our strategy has been a nonstandard quantum processing using decoherence. There might be
other approaches as quantum DAC has not been widely studied so far. Although its hardness
as a computational problem has been proved here, it is to be hoped that different physical
models and algorithmic strategies will be tried for more economical constructions.

\section*{Acknowledgments}
This work was supported by the Grant-in-Aid for Scientific Research from JSPS
(Grant No. 25871052).

\bibliographystyle{spmpsci-mod}
\bibliography{refs_qdac}
\newpage
\begin{center}
{\large Supplementary material of ``Quantum digital-to-analog conversion algorithm using decoherence''}\\
{Akira SaiToh}\\
Department of Computer Science and Engineering, Toyohashi University of Technology\\
1-1 Hibarigaoka, Tenpaku-cho, Toyohashi, Aichi 441-8580, Japan\\
Email: saitoh@sqcs.org
\end{center}
\newcommand{\iPi}{{\it\Pi}}
\section*{Details of the calculation of $S(\widetilde{\rho_{\rm c}}|\iPi^{\rm A})$}
Here is the details of the calculation of $S(\widetilde{\rho_{\rm c}}|\iPi^{\rm A})$
that has appeared in Sec.~4. As we have seen in the section, we choose $\iPi^{\rm A}
=\{\iPi^{\rm A}_1, \iPi^{\rm A}_2\}
=\{|\phi_1\rangle\langle\phi_1|\otimes|+\cdots+\rangle\langle+\cdots+|,
|\phi_2\rangle\langle\phi_2|\otimes|+\cdots+\rangle\langle+\cdots+|\}$
with $|\phi_1\rangle=\cos\theta|0\rangle+\sin\theta|1\rangle$ and
$|\phi_2\rangle=\sin\theta|0\rangle-\cos\theta|1\rangle$.
Let us write ${P'}^{\rm LR}=
\sum_{k|f_{\rm c}(k)=0}|f_{\rm c}(k)\rangle^{\rm L}\langle f_{\rm c}(k)|\otimes|k\rangle^{\rm R}\langle k|$
and
${P''}^{\rm LR}=
\sum_{k|f_{\rm c}(k)=2^{n-1}}|f_{\rm c}(k)\rangle^{\rm L}\langle f_{\rm c}(k)|\otimes|k\rangle^{\rm R}\langle k|$.
By definition, we have
\[
S(\widetilde{\rho_{\rm c}}|\iPi^{\rm A})=p_1S(\rho_1)+p_2S(\rho_2)
\]
with
\[
\left\{
\begin{array}{rl}
p_1&=\langle\iPi^{\rm A}_1\rangle=\frac{1}{2}\left[\frac{(\cos\theta+\sin\theta)^2}{2}+\cos^2\theta\right],\\
\rho_1&=\frac{\iPi^{\rm A}_1\widetilde{\rho_{\rm c}}\iPi^{\rm A}_1}{p_1}=
\begin{array}{rl}
\frac{1}{p_12^m}\biggl[&
{P'}^{\rm LR}\otimes\frac{(\cos\theta+\sin\theta)^2}{2}|\phi_1\rangle\langle\phi_1|\otimes|+\cdots+\rangle\langle+\cdots+|\\
&+{P''}^{\rm LR}\otimes\cos^2\theta|\phi_1\rangle\langle\phi_1|\otimes|+\cdots+\rangle\langle+\cdots+|
\biggr],
\end{array}
\end{array}
\right.
\]
and
\[
\left\{
\begin{array}{rl}
p_2&=\langle\iPi^{\rm A}_2\rangle=\frac{1}{2}\left[\frac{(\sin\theta-\cos\theta)^2}{2}+\sin^2\theta\right],\\
\rho_2&=\frac{\iPi^{\rm A}_2\widetilde{\rho_{\rm c}}\iPi^{\rm A}_2}{p_2}=
\begin{array}{rl}
\frac{1}{p_22^m}\biggr[&
{P'}^{\rm LR}\otimes\frac{(\sin\theta-\cos\theta)^2}{2}|\phi_2\rangle\langle\phi_2|\otimes|+\cdots+\rangle\langle+\cdots+|\\
&+{P''}^{\rm LR}\otimes\sin^2\theta|\phi_2\rangle\langle\phi_2|\otimes|+\cdots+\rangle\langle+\cdots+|
\biggl].
\end{array}
\end{array}
\right.
\]
Now $S(\widetilde{\rho_{\rm c}}|\iPi^{\rm A})$ is further calculated as follows.
\[
\begin{split}
S(\widetilde{\rho_{\rm c}}|\iPi^{\rm A})&=p_12^{m-1}\left[
-\frac{(\cos\theta+\sin\theta)^2}{p_12^{m+1}}\log_2\frac{(\cos\theta+\sin\theta)^2}{p_12^{m+1}}
-\frac{\cos^2\theta}{p_12^m}\log_2\frac{\cos^2\theta}{p_12^m}
\right]\\
&~~~\;+p_22^{m-1}\left[
-\frac{(\sin\theta-\cos\theta)^2}{p_22^{m+1}}\log_2\frac{(\sin\theta-\cos\theta)^2}{p_22^{m+1}}
-\frac{\sin^2\theta}{p_22^m}\log_2\frac{\sin^2\theta}{p_22^m}
\right]\\
&=-\frac{(\cos\theta+\sin\theta)^2}{4}\left[\log_2(\cos\theta+\sin\theta)^2-\log_2p_1-(m+1)\right]\\
&~~~\;-\frac{\cos^2\theta}{2}\left(\log_2\cos^2\theta-\log_2p_1-m\right)\\
&~~~\;-\frac{(\sin\theta-\cos\theta)^2}{4}\left[\log_2(\sin\theta-\cos\theta)^2-\log_2p_2-(m+1)\right]\\
&~~~\;-\frac{\sin^2\theta}{2}\left(\log_2\sin^2\theta-\log_2p_2-m\right)\\
&=-\frac{(\cos\theta+\sin\theta)^2}{4}\left[\log_2\frac{(\cos\theta+\sin\theta)^2}{2}-\log_2p_1-m\right]\\
&~~~\;-\frac{\cos^2\theta}{2}\left(\log_2\cos^2\theta-\log_2p_1-m\right)\\
&~~~\;-\frac{(\sin\theta-\cos\theta)^2}{4}\left[\log_2\frac{(\sin\theta-\cos\theta)^2}{2}-\log_2p_2-m\right]\\
&~~~\;-\frac{\sin^2\theta}{2}\left(\log_2\sin^2\theta-\log_2p_2-m\right)\\
&=m+p_1\log_2p_1+p_2\log_2p_2\\
&~~~\;\begin{array}{rl}-\frac{1}{2}\biggl[&
\frac{(\cos\theta+\sin\theta)^2}{2}\log_2\frac{(\cos\theta+\sin\theta)^2}{2}+\cos^2\theta\log_2\cos^2\theta\\
&+\frac{(\sin\theta-\cos\theta)^2}{2}\log_2\frac{(\sin\theta-\cos\theta)^2}{2}+\sin^2\theta\log_2\sin^2\theta
\biggr]\end{array}\\
&=m-H(p_1)+\frac{1}{2}H\left(\frac{(\cos\theta+\sin\theta)^2}{2}\right)+\frac{1}{2}H\left(\cos^2\theta\right).
\end{split}
\]

\end{document}